\documentclass[11pt, a4paper]{article}

\usepackage[utf8]{inputenc}
\usepackage[T1]{fontenc}
\usepackage{lmodern}
\usepackage{amsmath,amssymb,amsthm,mathtools}
\usepackage{bm}
\usepackage{booktabs}
\usepackage{graphicx}
\usepackage{microtype}
\usepackage{geometry}
\usepackage{algorithm}
\usepackage{algpseudocode}
\usepackage[dvipsnames]{xcolor}
\usepackage{natbib}
\setcitestyle{authoryear,round}
\usepackage[colorlinks=true,linkcolor=blue,citecolor=blue,urlcolor=blue]{hyperref}

\theoremstyle{plain}
\newtheorem{theorem}{Theorem}[section]
\newtheorem{proposition}[theorem]{Proposition}
\newtheorem{lemma}[theorem]{Lemma}
\theoremstyle{definition}

\newtheorem{example}[theorem]{Example}
\theoremstyle{remark}

\newcommand{\R}{\mathbb{R}}

\newcommand{\E}{\mathbb{E}}
\newcommand{\Var}{\operatorname{Var}}
\newcommand{\Cov}{\operatorname{Cov}}
\newcommand{\diag}{\operatorname{diag}}

\newcommand{\dd}{\,d}
\newcommand{\norm}[1]{\left\lVert #1 \right\rVert}
\newcommand{\abs}[1]{\left\lvert #1 \right\rvert}

\hypersetup{
    colorlinks=true,
    linkcolor=blue,
    filecolor=magenta,      
    urlcolor=cyan,
    citecolor=blue,
}

\title{Local Second-Order Geometry Induced by Deformation Maps}
\author{Maria Laura Battagliola \\ 
\small Department of Statistics, ITAM, Mexico City, Mexico}
\date{\today}

\begin{document}

\maketitle

\begin{abstract}
Spatial deformations offer a flexible route to nonstationary dependence by warping the coordinates of a stationary random field. While the exact induced covariance depends on the deformation map in its entirety, we show that its behavior in a neighborhood is approximated accurately by linearization. This produces a tangent covariance whose discrepancy from the true covariance we bound explicitly, and its Fourier transform yields a local spectrum in closed form. Building on this spectral description, we introduce a simulation scheme that generates a deformed Gaussian field in a neighborhood accounting for the local spectrum, so that the simulated field reproduces the finite dimensional tangent covariance by construction. For repeated sampling across many reference points, a truncated singular value decomposition compresses the space and frequency weights into a reusable form. We further apply the summaries based on the local Jacobian as an exploratory device for deformations estimated from images, using cardiac magnetic resonance data from the Automated Cardiac Diagnosis Challenge together with optical flow. The resulting local geometry exhibits differences across diagnostic groups through directional and anisotropic features of myocardial deformation that go beyond simple measures of local expansion or compression.
\end{abstract}

\section{Introduction}

Deformation maps provide a common representation of evolving complex objects. For instance, in medical imaging, registration estimates nonlinear transformations between images of the same patient at different times or those of different patients \citep{sotiras2013}. This, and many more examples, suggest that the deformation itself is an object of statistical interest.

In spatial statistics, deformations have played a central role in the construction of covariance models. Many frameworks operate under the  assumption of stationarity, meaning that the random field has a covariance structure that does not depend on the locations themselves, but only on the lag between them. This assumption is advantageous and widely used in practice because it leads to parsimonious covariance models, simplifies inference and prediction, and enables efficient computational tools for estimation and simulation (for a general overview, see \cite{cressie1993,stein1999}). However, when a nonlinear change of coordinates occurs, a stationary latent field generally becomes nonstationary in the original domain, since distances and directions are modified differently at different locations. This idea underlies the spatial deformation approach to nonstationary covariance modeling, where the observation domain is mapped to a latent space in which stationarity is more plausible \citep{sampson1992,damian2001,schmidt2003}. Related work has studied the identifiability and estimation of such deformations for isotropic Gaussian random fields, emphasizing the local geometric information carried by the deformation \citep{anderes2008}. 
Other work has estimated deformations from realizations by using local information in the observed field, either through small-scale quadratic variations \citep{guyon2000}, or through localized representations to recover deformation features of stationary processes \citep{clerc2003}.
Related covariance constructions allow the local dependence structure to vary across the domain, while still yielding valid nonstationary covariance models \citep{paciorek2006}. These studies motivate a local geometric analysis of nonstationary dependence.

A complementary perspective is provided by the spectral representation of spatial dependence. For stationary fields, the spectrum describes how variation is distributed over frequency, and changes in scale or direction correspond to transformations of this representation. For nonstationary processes, the idea of a spectrum that varies locally has a long history, beginning with evolutionary spectra for time series \citep{priestley1965}. Spatial spectral methods have also modeled nonstationarity through local stationarity, for instance by representing a spatial process as a kernel-weighted combination of stationary components whose spectral parameters vary across the domain \citep{fuentes2002}. Thus, nonstationarity can be studied complementarily through local frequency content. 

Simulation provides a further motivation for this perspective. Under a specified stationary covariance model, Gaussian random fields can often be simulated efficiently, especially on regular grids. This is the case, for instance, of circulant embedding \citep{wood1994, dietrich1997, graham2018}. For nonstationary fields, simulation is less straightforward because the dependence structure changes with location and the computational advantage of the stationary methods is generally lost. Existing approaches impose additional structure in different ways. One strategy is to assume approximate stationarity within moving spatial windows, estimate stationary covariance parameters locally, and then assemble these local estimates into a global model from which nonstationary fields can be simulated \citep{nychka2018}. Another strategy, proposed by \citet{fuglstad2015}, is to define the field through a stochastic partial differential equation whose parameters vary over space.
Moreover, \cite{emery2018} model anisotropy through a sum of random Fourier waves whose amplitudes vary with location according to a prescribed local spectral density.
When a global deformation map is available, \citet{kleiber2016} provides an efficient approach to nonstationary simulation by simulating a stationary Gaussian field in a latent space first, and then mapping it back to the desired warped domain.

We study how a smooth deformation changes
a random field’s local second-order structure, rather than treating the map only as a global simulation device.
A nonlinear deformation induces a valid but generally nonstationary covariance, and this exact covariance depends on the full deformation map. Locally, however, the action of the deformation is described by its Jacobian. We use this observation to derive a local covariance approximation, quantify the error introduced by the local linearization, and obtain the corresponding local spectrum. In this way, the deformation geometry provides the form of both covariance and spectrum.
Figure~\ref{fig:intro_fig} illustrates this idea. A nonlinear deformation applied to an isotropic Matérn field produces a deformed field with spatially varying local structure. Around a fixed location $s$, the exact local covariance induced by the deformation is well reproduced by the covariance obtained from the local linearization. The figure highlights the main point of the paper, namely that, even when the global deformation is complicated, its statistical effect can be understood locally through second-order quantities such as covariance and spectrum.

\begin{figure}[htbp]
    \centering
    \includegraphics[width=\linewidth]{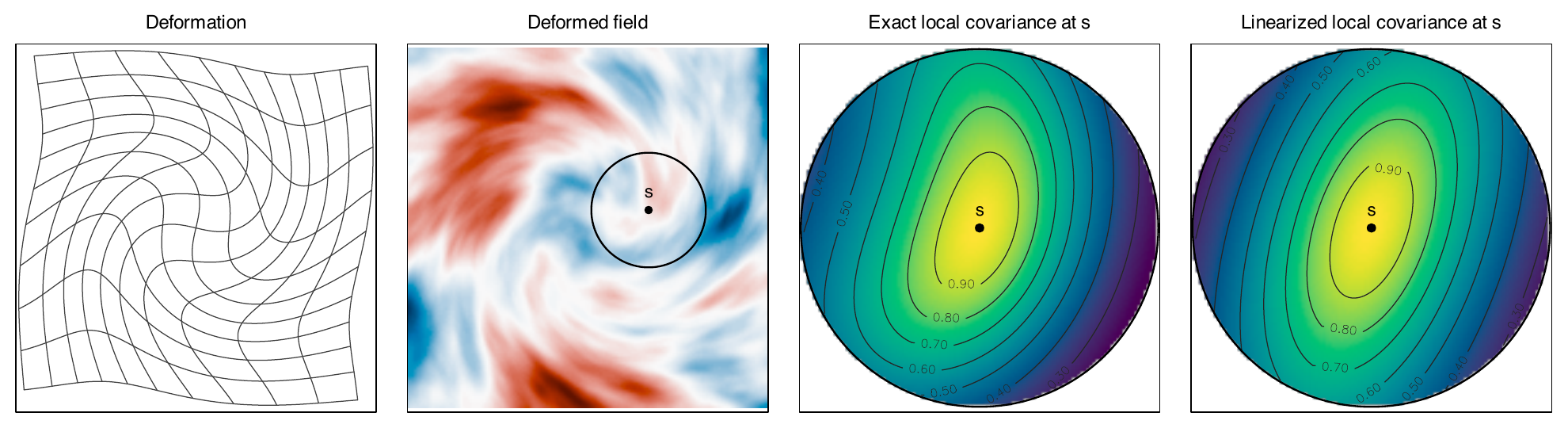}
    \caption{Illustration of deformed random field. The first panel shows the deformation, and the second panel shows such deformation applied to an isotropic Matérn random field, showing a location $s$ on the field. The third and fourth panels show the target local covariance and the approximated local covariance around $s$, respectively.}
    \label{fig:intro_fig}
\end{figure}

The same framework also offers a way to simulate deformed Gaussian fields locally. In particular, by knowing the local spectrum of the target field, we use a finite set of random Fourier components whose weights are determined by such local spectrum.
This produces a simulated local field whose covariance is exactly the finite-dimensional approximation of the linearized covariance.
When simulations are needed at many anchor locations, the space-frequency amplitude matrix can be compressed by a low-rank singular-value decomposition, yielding a compact representation that can be reused across repeated draws. The numerical results confirm that the covariance error follows the expected decomposition into finite-frequency and local-linearization components. They also show that the compression introduces only a small additional error, while the timing study identifies the regime in which the local simulator is computationally advantageous.

The same Jacobian-based summaries can be applied to estimated image deformations.
When the goal is to understand the local role of a deformation, we can begin with a direction-free geometry and let the estimated deformation determine the induced local anisotropy and orientation. This is especially relevant when the global deformation is not known a priori and must be estimated from the data. We illustrate this idea with cardiac magnetic resonance images from the Automated Cardiac Diagnosis Challenge (ACDC) data set \citep{bernard2018}, pairing the proposed local geometry with optical flow \citep{horn1981} to obtain local deformation estimates between images. The resulting summaries are then used as exploratory tools, showing that different pathology groups exhibit different localized patterns of myocardial deformation. In particular, the analysis suggests that these differences are not explained only by local expansion or compression, but also by spectral summaries of the local Jacobian in specific myocardial regions.

The rest of the paper is organized as follows. Section~\ref{sec:preliminaries} defines the deformation model and Fourier convention. Section~\ref{sec:local-properties} derives the exact deformed covariance, the local covariance error bound, and the local spectral warping formula. Section~\ref{sec:finite-frequency} develops the finite-frequency local simulator for a deformed stationary Gaussian random field. Section~\ref{sec:numerical-diagnostics} evaluates the approximation with numerical studies. Section~\ref{sec:acdc-application} applies the local geometry as an exploratory analysis for cardiac deformation, and Section~\ref{sec:discussion} discusses limitations and extensions.

\section{Preliminaries}
\label{sec:preliminaries}
Let $D\subset\R^d$ be compact, and let $X=\{X(u)\in \R :u\in\R^d\}$ be 
a zero-mean stationary random field 
with covariance
\begin{equation}
\label{eq:stat_cov_X}
     c_X(h)=\Cov (X(u),X(u+h)), \quad h\in \mathbb{R}^d.
\end{equation}
Equation~\eqref{eq:stat_cov_X} shows that the second-order properties of a stationary field only depend on lag $h=(h_1,\dots,h_d)^\top$, and it follows that $c_X(-h) = c_X(h)$. We can define the spectral density function $S_X$ as the Fourier transform of $c_X$, i.e. the following bijection holds
\begin{equation}
 \label{eq:spectral_density_X}
    \begin{aligned}
        S_X(k)&=\int_{\R^d}c_X(h)e^{-2\pi i k^\top h}\dd h, \quad  k\in \mathbb{R}^d,\\
        c_X(h)&=\int_{\R^d}S_X(k)e^{2\pi i k^\top h}\dd k, \quad  h\in \mathbb{R}^d,
    \end{aligned}
\end{equation}
where $k = (k_1,\dots,k_d)^\top\in \mathbb{R}^d$ is a frequency.
In particular, $S_X(-k) = S_X(k)$, and the variance of the field is  $c_X(0)=\int_{\R^d}S_X(k)\dd k$.

Now, let $T:D\to \R^d$ be a deterministic spatial deformation. We consider
$Y_T$, the deformation of $X$ under $T$, i.e.
\begin{equation}
    Y_T(s)=X(T(s)),\qquad s\in D.
    \label{eq:exact-deformed-field}
\end{equation}
For a fixed $s\in D$, we define the Jacobian of the transformation in $s$ as 
\begin{equation}
\label{eq:jacobian}
    J_T(s)=DT(s) = \begin{pmatrix}
\frac{\partial T_1}{\partial s_1} (s) & \dots & \frac{\partial T_1}{\partial s_d} (s)\\
\vdots & \ddots & \vdots\\
\frac{\partial T_d}{\partial s_1} (s) & \dots & \frac{\partial T_d}{\partial s_d} (s)
\end{pmatrix}.
\end{equation}
The Jacobian $J_T(s)$ describes the local first-order variation of the deformation at $s$, and the product $J_T(s)h$ is the directional derivative of $T$ at $s$ along the vector $h$, including the scale of $h$.
Moreover, consider $H_{T_r}(s)$, the Hessian matrix of the $r$th component $T_r$:
$$
H_{T_r}(s) = \begin{pmatrix}
\frac{\partial^2 T_r}{\partial s_1^2} (s) & \dots & \frac{\partial^2 T_r}{\partial s_1 \partial s_d} (s)\\
\vdots & \ddots & \vdots\\
\frac{\partial^2 T_r}{\partial s_d \partial s_1} (s) & \dots & \frac{\partial^2 T_r}{\partial s_d^2}(s)
\end{pmatrix}.
$$
Then, we consider the bilinear map $D^2T(s):\R^d\times\R^d\to\R^d$:
\begin{equation}
\label{eq:second_derivative}
    D^2T(s)(h,h) = \begin{pmatrix}
    h^\top H_{T_1}(s) h\\
    \vdots\\
    h^\top H_{T_d}(s) h
\end{pmatrix}.
\end{equation}
The term $D^2T(s)(h,h)$ measures the second-order directional variation of $T$ along $h$.
For the Taylor bounds below, we use 
\begin{equation}
    \norm{D^2T(s)}_{\operatorname{dir}}
    =
    \sup_{h\in \R^d}
    \norm{D^2T(s)\left (\frac{h}{\norm{h}},\frac{h}{\norm{h}} \right)},
    \label{eq:D2T-dir-norm}
\end{equation}
and we define the global curvature bound
\begin{equation}
    M_T
    =
    \sup_{s\in D}\norm{D^2T(s)}_{\operatorname{dir}}.
    \label{eq:global-curvature-bound}
\end{equation}
In what follows, we assume that $T\in C^2(D)$ and that $J_T(s)$ is nonsingular for every $s\in D$. Together with the compactness of $D$, this entails that $M_T$ in \eqref{eq:global-curvature-bound} is finite.

\section{Local properties of deformed stationary random fields}
\label{sec:local-properties}
Now we study the properties of the covariance and spectral density function of the deformed field $Y_T$ given $c_X$ and $S_X$ in \eqref{eq:spectral_density_X}. The exact deformation model changes the covariance by evaluating the covariance of $X$ at deformed lags. We show that the covariance of $Y_T$ remains valid for any
deterministic map $T$, but it generally loses stationarity. Then, to obtain a local spectral description, we fix a
location $s$ and use the first order approximation of the deformed lag $(T(s+h)-T(s))$.

First, we show that $Y_T$ has a valid covariance with a preliminary lemma.
\begin{lemma}[Validity of the exact deformed covariance]
\label{th:valid_cov}
Let $X$ be a zero-mean stationary random field with covariance $c_X$. Let $T:D\to\R^d$ be deterministic and define $Y_T(s)=X(T(s))$. Then, $Y_T$ is a valid zero-mean random field on $D$ with covariance
\begin{equation}
    C_T(s,t)=\Cov(Y_T(s),Y_T(t))
    =
    c_X(T(t)-T(s)).
    \label{eq:exact-deformed-covariance}
\end{equation}
Moreover, $C_T$ is positive semidefinite.
\end{lemma}
\begin{proof}
    By stationarity of $X$, for any $s,t\in D$,
$$\Cov(Y_T(s),Y_T(t))=\Cov(X(T(s)),X(T(t)))=c_X(T(t)-T(s)).$$
    Moreover, for any finite collection $s_1,\dots,s_n\in D$, $T(s_1),\dots,T(s_n) \in \R^d$.
    Since $c_X$ is a valid
    stationary covariance on $\R^d$, the covariances
    $$C_T(s_i,s_j) =c_X(T(s_j)-T(s_i)),\qquad i,j=1,\dots,n,$$
    form a positive semidefinite covariance matrix.
    Therefore $C_T$ is a valid
    covariance function on $D$.
\end{proof}

The result above ensures that the deformation on $X$ leads to the valid covariance $C_T$ for $Y_T$. Notice that such covariance is not stationary. In particular, for any $s_1,t_1,s_2,t_2\in D$ such that $(s_1-t_1)=(s_2-t_2)$, leading to $c_X(t_1-s_1)= c_X(t_2-s_2)$, it might be that $c_X(T(t_1)-T(s_1))\ne c_X(T(t_2)-T(s_2))$. Moreover, the marginal variance of $Y$ is the same as that of $X$, i.e. $C_T(s,s) = c_X(0)$.

Our aim is to verify the local second order properties of $Y_T$ when $(T(s+h)-T(s))$ is replaced by the local linear approximation $J_T(s)h$. The following theorem shows that, under regularity conditions on $T$ and $c_X$, the covariance error is controlled by the global curvature bound \eqref{eq:global-curvature-bound} and the magnitude of lag $h$.
\begin{theorem}[Local covariance linearization]
\label{thm:local-cov-linearization}
Given the map $T:D\to\R^d$, assume $T\in C^2(D)$. Moreover consider the stationary zero-mean random field $X$, whose covariance $c_X$ is Lipschitz continuous with constant $L_c$:
$$
    \abs{c_X(u)-c_X(v)}
    \le L_c\norm{u-v},
    \qquad u,v\in\R^d.
$$
For $s\in D$, consider $D_s=\{h\in\R^d:\ s+\theta h\in D, \;\forall 0\le\theta\le1\}$, namely the set of all the lags $h\in\R^d$ such that the line segment joining $s$ and $(s+h)$
is contained in $D$. Then, for $h\in D_s$,
we have
\begin{equation}
    C_T(s,s+h)
    =
    c_X(J_T(s)h)+R_c(s,h),
    \label{eq:local-cov-linearization}
\end{equation}
where
\begin{equation}
    \abs{R_c(s,h)}
    \le
    \frac{L_c}{2}M_T\norm{h}^2.
    \label{eq:local-cov-bound}
\end{equation}
\end{theorem}
\begin{proof}
Consider a fixed $s \in D$. Then, taking $h \in D_s$ and applying Taylor's theorem with integral remainder
gives
\begin{equation}
\label{eq:first_order_taylor}
     T(s+h)
    =
    T(s)+J_T(s)h+R_T(s,h),
\end{equation}
where $R_T(s,h)
    =
    \int_0^1
    (1-\theta)D^2T(s+\theta h)(h,h)\dd\theta$.
Since the segment from $s$ to $(s+h)$ lies in $D$ and
$\norm{D^2T(t)}_{\operatorname{dir}}\le M_T$ for any $t\in D$,
$$
    \norm{R_T(s,h)}
    \le
    \int_0^1(1-\theta)M_T\norm{h}^2\dd\theta
    =
    \frac12 M_T\norm{h}^2,
$$
where $\norm{h}^2$ appears as a scaling constant for $M_T$, which is defined on unit directions.
Therefore
$$
    C_T(s,s+h)
    =
    c_X(T(s+h)-T(s))
    =
    c_X(J_T(s)h+R_T(s,h)).
$$
Define $R_c(s,h)
    =
    c_X(J_T(s)h+R_T(s,h))-c_X(J_T(s)h)$.
By the Lipschitz condition on $c_X$,
$$
    \abs{R_c(s,h)}
    \le
    L_c\norm{R_T(s,h)}
    \le
    \frac{L_c}{2}M_T\norm{h}^2.
$$
\end{proof}

Theorem~\ref{thm:local-cov-linearization} shows that the target covariance $C_T$ can be linearized with an error that is bounded above. In particular, it is guaranteed that for small lags around a location, this error is controlled by the supremum of the curvature of the deformation. Thus, the more nonlinear the deformation is, the less accurately it is represented by its local linear approximation. As a result, the covariance error can increase even for small values of $\norm{h}$.

Finally, we study the properties of the deformation on the local spectral density function.
\begin{theorem}[First-order spectral warping and variance preservation]
\label{thm:spectral-warping} 
Fix $s\in D$ and assume that $J_T(s)$ is nonsingular. The spectral density of the locally linearized covariance $ c_X(J_T(s)h)$ is
\begin{equation}
    S_T^{\operatorname{loc}}(k;s)
    =
    \frac{1}{\abs{\det J_T(s)}}
    S_X(J_T(s)^{-\top}k).
    \label{eq:local-spectral-density}
\end{equation}
Moreover,
$$
\int_{\R^d}S_T^{\operatorname{loc}}(k;s)\dd k
    =
    c_X(0).
$$
\end{theorem}
\begin{proof}
Given location $s \in D$, lag $h \in \R^d$ and deformation $T: D \to \R^d$, the locally linearized covariance around $s$ is
$c_X(J_T(s)h)$.
The corresponding local spectral density is its Fourier transform with
respect to $h$:
$$
    S_T^{\operatorname{loc}}(k;s)
    =
    \int_{\R^d}
    c_X(J_T(s)h)
    e^{-2\pi i k^\top h}
    \dd h.
$$
Using the change of
variables $u=J_T(s)h$, then $h=J_T(s)^{-1}u$ and $\dd h=\abs{\det J_T(s)}^{-1}\dd u$. Since $k^\top h=(J_T(s)^{-\top}k)^\top u$, we get
\begin{align*}
    S_T^{\operatorname{loc}}(k;s)
    =
    \abs{\det J_T(s)}^{-1}
    \int_{\R^d}
    c_X(u)
    e^{-2\pi i (J_T(s)^{-\top}k)^\top u}
    \dd u 
    =
    \frac{1}{\abs{\det J_T(s)}}
    S_X(J_T(s)^{-\top}k).
\end{align*}
Integrating over $k$ and using the change of variables
$v=J_T(s)^{-\top}k$, so that
$k=J_T(s)^\top v$ and
$\dd k=\abs{\det J_T(s)}\dd v$, gives
$$
    \int_{\R^d}S_T^{\operatorname{loc}}(k;s)\dd k
    =
    \int_{\R^d}
    \frac{1}{\abs{\det J_T(s)}}S_X(J_T(s)^{-\top}k)\dd k
    =
    \int_{\R^d}S_X(v)\dd v
    =
    c_X(0).
$$
\end{proof}
Theorem~\ref{thm:spectral-warping} gives us a local spectrum in closed form, which recovers the local variance of the deformed field $Y$. 

Now, we present the local covariance and spectrum of some notable random field examples.
\begin{example}[Affine deformation]
\label{ex:affine}
Suppose
$$
    T(s)=As+b,
$$
where $A\in\R^{d\times d}$ is nonsingular and $b\in\R^d$. Then $J_T(s)=A$ for all $s\in D$, and the exact covariance is
$$
    C_T(s,t)
    =
    c_X(A(t-s)).
$$
Thus, the deformed field remains stationary. The local spectral formula is exact in this case:
$$
    S_T^{\operatorname{loc}}(k;s)
    =
    \abs{\det A}^{-1}S_X(A^{-\top}k),
$$
which does not depend on $s$. 
\end{example}

\begin{example}[Matérn local spectrum and anisotropy]
\label{ex:Matern}
Suppose the latent covariance is Matérn,
$$
   c_X(h)
    =
    \sigma^2
    \frac{2^{1-\nu}}{\Gamma(\nu)}
    \left(
        \frac{\sqrt{2\nu}}{\rho}\norm{h}
    \right)^\nu
    K_\nu\left(
        \frac{\sqrt{2\nu}}{\rho}\norm{h}
    \right),
    \qquad
    \nu>0,\quad \rho>0.
$$
Then, the corresponding spectral density can be written as
$$
   S_X(k)
    =
    (2\pi)^d\sigma^2
    \frac{\Gamma(\nu+d/2)(2\nu)^\nu}
    {\Gamma(\nu)\pi^{d/2}\rho^{2\nu}}
    \left(
        \frac{2\nu}{\rho^2}
        +
        4\pi^2\norm{k}^2
    \right)^{-(\nu+d/2)}.
$$
Using the result of Theorem~\ref{thm:spectral-warping}, we obtain
$$
    S_T^{\operatorname{loc}}(k;s)
    =
    \frac{(2\pi)^d\sigma^2}{\abs{\det J_T(s)}}
    \frac{\Gamma(\nu+d/2)(2\nu)^\nu}
    {\Gamma(\nu)\pi^{d/2}\rho^{2\nu}}
    \left(
        \frac{2\nu}{\rho^2}
        +
        4\pi^2
        k^\top
        J_T(s)^{-1}J_T(s)^{-\top}
        k
    \right)^{-(\nu+d/2)}.
$$
Thus, the deformation turns an isotropic Matérn spectrum into a locally anisotropic Matérn spectrum. The frequency-domain anisotropy is governed by $J_T(s)^{-1}J_T(s)^{-\top}$.
The corresponding first-order local covariance is
$$
   c_X(J_T(s)h)
    =
    \sigma^2
    \frac{2^{1-\nu}}{\Gamma(\nu)}
    \left(
        \frac{\sqrt{2\nu}}{\rho}
        \sqrt{h^\top G(s)h}
    \right)^\nu
    K_\nu\left(
        \frac{\sqrt{2\nu}}{\rho}
        \sqrt{h^\top G(s)h}
    \right).
$$
where
\begin{equation}
\label{eq:def_Matern_cov}
    G(s)=J_T(s)^\top J_T(s).
\end{equation}
Therefore, $G(s)$ controls the local covariance level curves in space, while $G(s)^{-1}$ controls the level curves of the spectrum in frequency.
The local Jacobian changes the effective range and orientation of the Matérn field, but it does not change the Matérn smoothness parameter $\nu$. By Theorem~\ref{thm:spectral-warping}, the determinant factor preserves the nominal marginal variance $\sigma^2$ locally. Notice that Theorem~\ref{thm:local-cov-linearization} can be applied to the Matérn covariance when $\nu \geq 1/2$.
\end{example}

\section{Local spectral simulation of Gaussian fields}
\label{sec:finite-frequency}

The previous section describes the spectral behavior that a deformed stationary field should exhibit locally. In this section, we show that the derived local spectrum can be used to simulate Gaussian fields under known deformation $T$ in an efficient way. Notice that the derived spectrum does not define the global spectrum over the entire domain. Thus, we will establish guarantees on the error over local covariance structure.

\subsection{Finite-frequency local-spectrum field}

Consider location $s \in D$ and frequency $k \in \R^d$. As discussed previously, the local spectrum $S_T^{\operatorname{loc}}(k;s)$ describes how variance of the deformed random field is distributed across frequencies near location $s$. To simulate from this local spectral structure, we use:
\begin{equation}
\label{eq:filter}
     F_T(k,s)=\left [S_T^{\operatorname{loc}}(k;s)\right ]^{1/2}.
\end{equation}
We call this the local spectral amplitude, and it determines how strongly a random Fourier component at frequency $k$ contributes to the field at location $s$. 
Because $S_T^{\operatorname{loc}}(k;s)$ is a valid spectrum thanks to Theorem~\ref{thm:spectral-warping}, $F_T(k,s)$ is guaranteed to be real and nonnegative. 

Now, let $K\subset\mathbb R^d$ be a compact symmetric frequency window used to
approximate the full spectral domain $\mathbb R^d$. For simplicity, we take a
common cutoff $\kappa>0$ in all coordinate directions and define the symmetric
hypercube
$$
K=[-\kappa,\kappa]^d
=
\left\{
k=(k_1,\ldots,k_d)\in\mathbb R^d:
|k_r|\le \kappa,\ r=1,\ldots,d
\right\}.
$$
The construction extends directly to coordinate-specific cutoffs.
Let $K_L\subset K$ be a regular symmetric frequency grid with spacing
$\Delta k$ in each frequency coordinate, and write $\Delta_K=(\Delta k)^d$
for the frequency-cell volume. Here symmetric means that if $k\in K_L$, then
$-k\in K_L$, where $-k=(-k_1,\ldots,-k_d)$. Because the local spectrum is even in frequency, the contributions from $k$ and
$-k$ can be paired. We choose a half-grid $K_L^+\subset K_L\setminus\{0\}$ such
that $K_L =
\{0\}\,\cup\,K_L^+\,\cup\,-K_L^+$, where $-K_L^+=\{-k:k\in K_L^+\}$ denotes the reflection of $K_L^+$ through the
origin. 
Then, we use the grid $K_L^+ = \{ k_1,\dots,k_L\}$, and we tread the zero frequency separately.

Let $Z_0,U_1,V_1, \dots,U_L,V_L \stackrel{iid}{\sim} \mathcal{N}(0,1)$.
For a fixed anchor location $s\in D$, choose a local radius $\rho>0$ such that
$B(s,\rho) = \{r\in D: \; \norm{r-s} \leq \rho \}$ is the ball centered in $s$ with radius $\rho$. Equivalently, we consider lags $\mathcal H_\rho(s)
=
\{h\in\mathbb R^d:\ (s+h)\in B(s,\rho)\}$.
We define a finite-frequency local-spectrum field on
$B(s,\rho)$ by freezing the local spectral amplitude at the anchor $s$:
\begin{equation}
\label{eq:simulated_field}
    Y_L^{(s)}(r)
=
\sqrt{\Delta_K}\,F_T(0,s)Z_0
+
\sqrt{2\Delta_K}
\sum_{\ell=1}^L
F_T(k_\ell,s)
\left[
U_\ell\cos(2\pi k_\ell^\top r)
+
V_\ell\sin(2\pi k_\ell^\top r)
\right],
\quad r\in B(s,\rho).
\end{equation}
The field in \eqref{eq:simulated_field} is a finite-frequency realization
of the linearized Gaussian model at $s$.
A related continuous spectral construction is proposed by
\citet{emery2018}, where nonstationarity is specified directly through
a location-dependent spectral density rather than through a deformation
map. The field is approximated by a finite sum of cosine waves with
random frequencies and phases, weighted by the square root of the local
spectral density relative to the frequency-sampling density.

Since throughout the neighborhood $B(s,\rho)$,
the spectrum is frozen at $S_T^{\operatorname{loc}}(\cdot;s)$, $Y_L^{(s)}$ is a local field whose
second-order structure reproduces the finite-frequency approximation to the
linearized covariance at $s$. This construction is justified by the fact that its second order properties are finite dimensional representations of the locally linearized quantities discussed in Section~\ref{sec:local-properties}. This is shown in the following result.
\begin{proposition}[Local moments of the finite-frequency field]
\label{prop:finite-frequency-covariance}
Fix $s\in D$, and let $\rho>0$ be such that $B(s,\rho)\subset D$. For any
$h \in\mathcal H_\rho(s)$, the field $Y_L^{(s)}$ is zero-mean with covariance and variance
\begin{equation}
\label{eq:cov_YL}
    C_{L,s}(h)=\Cov\left(Y_L^{(s)}(s),Y_L^{(s)}(s+h)\right)
=
\Delta_K S_T^{\operatorname{loc}}(0;s)
+
2\Delta_K
\sum_{\ell=1}^L
S_T^{\operatorname{loc}}(k_\ell;s)
\cos(2\pi k_\ell^\top h),
\end{equation}
and
\begin{equation}
\label{eq:var_YL}
v_{L,s}=\Var(Y_L^{(s)}(s))
=
\Delta_K S_T^{\operatorname{loc}}(0;s)
+
2\Delta_K
\sum_{\ell=1}^L
S_T^{\operatorname{loc}}(k_\ell;s),
\end{equation}
respectively.
\end{proposition}
\begin{proof}[Proof]
Since $Y_L^{(s)}(r)$ is a linear combination of zero-mean Gaussian random variables,
it is zero-mean and Gaussian. 
We now compute the covariance. In particular, for each $\ell=1,\dots,L$, write $\theta_\ell(r)=2\pi k_\ell^\top r$, for $r \in B(s,\rho)$. Then, for $h$ such that $(s+h)\in B(s,\rho)$, we have
\begin{align*}
    \Cov(Y_L^{(s)}(s),Y_L^{(s)}(s+h)) & = \E\left[
\sqrt{\Delta_K}F_T(0,s)Z_0
\sqrt{\Delta_K}F_T(0,s)Z_0
\right] \\
& + 2\Delta_K \sum_{\ell=1}^LF_T(k_\ell,s)^2
\E\left[
(U_\ell\cos\theta_\ell(s)+V_\ell\sin\theta_\ell(s))
(U_\ell\cos\theta_\ell(s+h)+V_\ell\sin\theta_\ell(s+h))
\right]\\
&=  \Delta_K F_T(0,s)^2 + 2\Delta_K \sum_{\ell=1}^LF_T(k_\ell,s)^2 \cos(\theta_\ell(s) - \theta_\ell(s + h)),
\end{align*}
using the independence of the Gaussian coefficients and trigonometric identities. Since $F_T(0,s)^2 = S^\text{loc}_T(0;s)$ and $F_T(k_\ell,s)^2 = S^\text{loc}_T(k_\ell;s)$ by definition, and $\cos(\theta_\ell(s) - \theta_\ell(s + h)) = \cos(2\pi k_\ell^\top h)$, we obtain
$$
 \Cov(Y_L^{(s)}(s),Y_L^{(s)}(s+h)) = \Delta_KS^\text{loc}_T(0;s) + 2\Delta_K\sum_{\ell=1}^LS^\text{loc}_T(k_\ell;s)\cos(2\pi k_\ell^\top h).
$$
Finally, taking $h=0$, we obtain
$$
 \Var(Y_L^{(s)}) = \Delta_KS^\text{loc}_T(0;s) + 2\Delta_K\sum_{\ell=1}^LS^\text{loc}_T(k_\ell;s).
$$
\end{proof}
Therefore, for each fixed anchor $s\in D$, the covariance of $Y_L^{(s)}$ is the discrete inverse Fourier transform of the local spectrum
$S_T^{\operatorname{loc}}(\cdot;s)$. Consequently, \eqref{eq:cov_YL} is the finite-frequency approximation to
the tangent covariance $c_X(J_T(s)h)$, and the
finite-frequency approximation error is
$$
\varepsilon_L(s,h)
=
\left|
C_{L,s}(h)-c_X(J_T(s)h)
\right|.
$$
The local error between the covariance in \eqref{eq:cov_YL} and the exact deformed covariance at $h \in \mathcal{H}_\rho(s)$ is such that
\begin{equation}
\label{eq:error_cov}
    \left|
C_{L,s}(h)-C_T(s,s+h)
\right|
\le
\varepsilon_L(s,h)
+
\frac{L_c}{2}M_T\rho^2,
\end{equation}
where we used the local linearization bound from Theorem~\ref{thm:local-cov-linearization}.

With the guarantee of properly reproducing the local covariance structure, \eqref{eq:simulated_field} can be used to simulate the target Gaussian random field around an anchor location. Such realizations help understand the statistical structure induced by the deformation, including the orientation and the range of local anisotropy, but also for inference, such as confidence surfaces for the functionals of interest.

\subsection{Implementation}
\label{sec:implementation}
In this section, we discuss practical aspects of the proposed simulation method. Algorithm~\ref{alg:sim} gives the procedure for generating one realization of the random field under the linearized covariance model in a neighborhood centered at the location of interest, with a prespecified radius.

\begin{algorithm}[htbp]
\caption{Local finite-frequency tangent simulation at anchor $s$}
\begin{algorithmic}
\Require  $S_X$, $T$, $s\in D$, $\rho>0$, $\{r_1,\ldots,r_Q \}\in B(s,\rho)$,
 $K_L^+=\{k_1,\ldots,k_L\}$, $\Delta_K$.
\Ensure A realization $Y_1,\ldots,Y_Q$, where $Y_q$ approximates $Y_L^{(s)}(r_q)$.
\State Compute $J \gets J_T(s)$ and $d_J \gets |\det J|$.
\State Set
$
F_0 \gets \left [d_J^{-1}S_X(0)\right]^{1/2}.
$

\For{$\ell=1,\ldots,L$}
    \State Set $q_\ell \gets J^{-\top}k_\ell$.
    \State Set
    $
    F_\ell \gets \left [d_J^{-1}S_X(q_\ell)\right]^{1/2}.
    $
\EndFor

\State Draw 
$
Z_0,U_1,V_1,\ldots,U_L,V_L
\stackrel{\mathrm{iid}}{\sim} \mathcal N(0,1).
$

\For{$q=1,\ldots,Q$}
    \State Initialize
    $
    Y_q \gets \sqrt{\Delta_K}\,F_0Z_0.
    $
    \For{$\ell=1,\ldots,L$}
        \State Update
        $
        Y_q
        \gets
        Y_q
        +
        \sqrt{2\Delta_K}\,F_\ell
        \left[
        U_\ell\cos(2\pi k_\ell^\top r_q)
        +
        V_\ell\sin(2\pi k_\ell^\top r_q)
        \right].
        $
    \EndFor
\EndFor

\State \Return $Y_1,\ldots,Y_Q$.
\end{algorithmic}
\label{alg:sim}
\end{algorithm}

For a location $s\in D$, the marginal variance of the finite-frequency tangent field is $v_{L,s}$ in \eqref{eq:var_YL}. In the continuous-frequency model, this quantity corresponds to $c_X(0)$. In
practice, frequency truncation and discretization may introduce an error in the variance. However, since the covariance of the latent stationary field $X$ is available, we can enforce the true marginal variance using
$$
\widetilde{Y}_{L}^{(s)}(r)
=
\left [\frac{c_X(0)}{v_{L,s}}\right]^{1/2}Y_L^{(s)}(r),
\quad r\in B(s,\rho).
$$
Notice that this correction fixes the local marginal variance at every point in the neighborhood, but it also
rescales the off-diagonal covariances by the factor $c_X(0)/v_{L,s}$. Therefore it
is most appropriate when $v_{L,s}$ is already close to $c_X(0)$, so that the
correction is small.

Moreover, repeating Algorithm~\ref{alg:sim} gives independent
realizations of the finite-frequency local field at $s\in D$.
When simulations are needed at several anchor locations
$\{s_1,\ldots,s_N\}$, the local amplitudes can be collected in the matrix $\mathcal F$, where $[\mathcal F]_{i\ell}=F_T(k_\ell,s_i)$ for $i=1,\ldots,N$ and $\ell=1,\ldots,L$.
Then, a truncated singular
value decomposition (SVD) of rank $M$, i.e. $\mathcal F
\approx
\mathcal F_M
=
U_M\Sigma_MV_M^\top$, provides a compact representation of the filter, where $\Sigma_M = \diag(\sigma_1, \dots, \sigma_M)$ is the diagonal matrix whose entries are the ordered singular values. Entrywise, 
$$F_T(k_\ell,s_i)
\approx
F_M(k_\ell,s_i)
=
\sum_{m=1}^M
\sigma_m [U]_{im}[V]_{\ell m}.$$
For each component $m$, the factor $[V]_{\ell m}$ gives the value of the corresponding frequency mode at $k_\ell$, the factor $[U]_{im}$ gives the spatial loading at anchor location $s_i$, and the singular value $\sigma_m$ gives the overall weight of that component.
This step reduces storage from $O(NL)$ to $O(M(N+L))$, and is convenient when
$M\ll\min(N,L)$. 
This provides a reusable representation for repeated local
simulation across many anchors.
Specifically, once the low-rank representation has been computed, additional simulations are
obtained by keeping $U_M$, $\Sigma_M$, and $V_M$ fixed and drawing new
Gaussian coefficients in Algorithm~\ref{alg:sim}. Thus, for
anchor $s_i$, the algorithm is applied with $F_M(k_\ell,s_i)$ in place of
$F_T(k_\ell,s_i)$. The covariance induced by the compressed simulator at $h \in \mathcal{H}_\rho(s)$ is therefore
\begin{equation}
\label{eq:SVD_cov}
 C_{L,M,s_i}(h)
=
\Delta_K F_T(0,s_i)^2
+
2\Delta_K
\sum_{\ell=1}^L
F_M(k_\ell,s_i)^2
\cos(2\pi k_\ell^\top h).
\end{equation}

Finally, a locally tuned equispaced Fourier grid can be used to
construct $ K_L^+$. Two quantities must be selected: the
frequency spacing $\Delta_K$, which determines the period of the
Fourier approximation and hence controls aliasing, and the univariate frequency
cutoff $\kappa = L \,\Delta k$, which controls spectral truncation. \citet{barnett2024} propose first choosing the largest frequency
spacing compatible with a prescribed aliasing tolerance, and then
increasing the grid extent until the truncation error due to the
omitted spectral tail is of approximately the same magnitude as, or
smaller than, the aliasing error. Consequently, the number of frequencies varies
with the local scale and anisotropy at each location. When a common
grid is required for several anchor locations, we propose using the smallest
local spacing and the largest local cutoff.

\section{Numerical diagnostics}
\label{sec:numerical-diagnostics}

In this section we use numerical results to assess the finite-dimensional covariance error, the covariance error introduced by the low-rank SVD approximation, as well as the speed up gained.

\subsection{Simulation cases}
\label{subsec:simulation-cases}

We consider $D=[0,1]^2$ and the regular grid
$D_n=\{(i/(n-1),j/(n-1)):i,j=0,\ldots,n-1\}$ with $n=128$. The latent process $X$ is a zero-mean stationary Gaussian field, and the deformed process is $Y_T(s)=X(T(s))$ for $s\in D$, with known $T:D \to \R^2$.
For covariance diagnostics, the target is the exact deformed covariance, i.e. $C_T(s,t)=c_X(T(t)-T(s))$. For $c_X$ we use the Matérn model 
\begin{equation}
\label{eq:cX_model_sim}
    c_X(h) = \sigma^2
\frac{2^{1-\nu}}{\Gamma(\nu)}
\left(
\sqrt{2\nu}\,\|L^{-1}h\|
\right)^\nu
K_\nu\left(
\sqrt{2\nu}\,\|L^{-1}h\|
\right),
\end{equation}
with $\sigma^2=1$. Moreover, for $L = \begin{pmatrix}
\rho & 0\\
0 & \rho
\end{pmatrix}$, with $\rho=0.15$, we obtain the isotropic model, and for $L
=
R_\alpha
\begin{pmatrix}
\rho_1 & 0\\
0 & \rho_2
\end{pmatrix},$ with $\rho_1=0.30$ and $\rho_2=0.08$, and where $R_\alpha = \begin{pmatrix}
    \cos(\alpha) & -\sin(\alpha)\\
    \sin(\alpha) & \cos(\alpha)
\end{pmatrix}$ is a rotation matrix with orientation $\alpha=\pi/4$. We will consider both model specifications. 

The corresponding spectrum to \eqref{eq:cX_model_sim} is
$$
    S_X(k)=
|\det L|\,4\pi\sigma^2\nu(2\nu)^\nu
\left(2\nu+4\pi^2\norm{L^\top k}^2\right)^{-(\nu+1)}.
$$
In the numerical
results reported below, we check for different smoothness values, namely $\nu\in\{0.5,1,1.5\}$.

The same four deformation maps are used for both Matérn baselines.
The affine stretch is
\[
T_{\mathrm{aff}}(s)
=
c + A(s-c),
\qquad
A =
R_\theta
\begin{pmatrix}
a_x & 0\\
0 & a_y
\end{pmatrix},
\]
with $c=(1/2,1/2)^\top$, $a_x=1.25$, $a_y=0.80$, and $\theta=0$.
The transport shear is
\[
T_{\mathrm{sh}}(s_1,s_2)
=
\left(
s_1-\tau A\sin(2\pi s_2),
s_2
\right),
\]
with $\tau=1$ and $A=0.18$. The radial lens is
\[
T_{\mathrm{lens}}(s)
=
c+q(s)(s-c),
\qquad
q(s)
=
1+A\exp\left(-\frac{\|s-c\|^2}{r^2}\right),
\]
with $c=(1/2,1/2)^\top$, $A=0.75$, and $r=0.30$. Finally, the vortex is
\[
T_{\mathrm{vor}}(s)
=
c+R_{\theta(s)}(s-c),
\qquad
\theta(s)
=
\omega
\exp\left (-\frac{\|s-c\|^2}{r^2}\right),
\]
with $c=(1/2,1/2)^\top$, $\omega=1.8$, and $r=0.35$.

For the affine case the Jacobian is constant, as shown in Example~\ref{ex:affine}, and the local
linearization is exact. The other examples are more challenging non-linear transformations. The shear and vortex are area-preserving since $\det J_T(s)=1$ for any $s\in D$, while the lens changes local area.

\subsection{Covariance and compression errors}
\label{subsec:covariance-errors}

The first diagnostic separates two sources of covariance error. For the analysis, we consider anchor locations
$$
\mathcal A
= \left\{ \left(\frac{j-1/2}{n_A},\frac{k-1/2}{n_A} \right) \in D, \;  j,k=1,\ldots,n_A \right\},
$$
with $n_A=8$, with $|\mathcal{A}|= 64$.
For each anchor $s_0 \in \mathcal A$ and for each
radius $\rho\in\{0.04,0.08,0.12,0.16,0.20\}$, we $h \in B(s_i, \rho)$ and compute
\begin{align*}
     e_{\mathrm{freq},\rho}(s_0,h) & =|C_{L,s_i}(h)-c_X(J_T(s_0)h)|,\\
      e_{\mathrm{lin},\rho}(s_0,h)&=|c_X(J_T(s_0)h)-C_T(s_0,s_0+h)|.
\end{align*}
In particular, for fixed $s_0 \in \mathcal{A}$ and $\rho >0$, $e_{\mathrm{freq},\rho}(s_0,h)$ measures the contribution to the error made in the finite frequency representation, while $e_{\mathrm{lin},\rho}(s_0,h)$ measures the linearization error. They both play a role in the global error between the target covariance and that of the simulated local field for $h\in \mathcal{H}_\rho(s_0)$:
$$
 |C_{L,s_0}(h)-C_T(s_0,s_0+h)| \leq  e_{\mathrm{freq},\rho}(s_0,h) + e_{\mathrm{lin},\rho}(s_0,h).
$$
This is the numerical counterpart to bound \eqref{eq:error_cov}. We know that $e_{\mathrm{freq},\rho}(s_0,h)$ depends on the discratization of the grid, while $e_{\mathrm{lin},\rho}(s_0,h)$ depends quadratically on $\rho$.

For each error type, we report the mean absolute error (MAE) over the anchor points in $\mathcal{A}$, i.e.
$$
\text{MAE}_{\text{type}}(\rho) = \frac{1}{|\mathcal{A}|} \sum_{s_0 \in \mathcal A} \frac{1}{|\mathcal{H}_\rho(s_0)|} \sum_{h\in \mathcal{H}_\rho(s_0)}e_{\text{type},\rho}(s_0,h), \quad \text{type} \in \{\text{freq}, \text{lin} \}.
$$

The results are shown in Figure~\ref{fig:cov_error_rho}. First, the affine linearization error is at numerical zero in all cases, as expected, and is not reported in the graph for readability. For the nonlinear deformations, the linearization component increases with $\rho$, while the finite-frequency component decreases in the reported diagnostics. This behavior is consistent with the theoretical linearization bound, which is of order $\rho^2$, whereas the finite-frequency component reflects the accuracy of the spectral truncation and discretization used in the approximation. Moreover, for all nonlinear transformations and for both the isotropic and anisotropic cases, the finite-frequency error decreases as the latent field becomes smoother. This is because larger values of $\nu$ correspond to less high-frequency spectral mass, so the finite-frequency representation reproduces the local covariance more accurately than in rougher cases. A similar effect is observed in the anisotropic setting, where $e_{\mathrm{freq},\rho}$ is generally smaller than in the isotropic counterpart. On the other hand, the linearization error is generally smaller in the isotropic cases than in the anisotropic ones. This is expected because anisotropic covariances are more sensitive to the direction of the deformed lag, making the nonlinear deformation harder to capture with a first-order local linearization.

\begin{figure}[htbp]
    \centering
    \includegraphics[width=\linewidth]{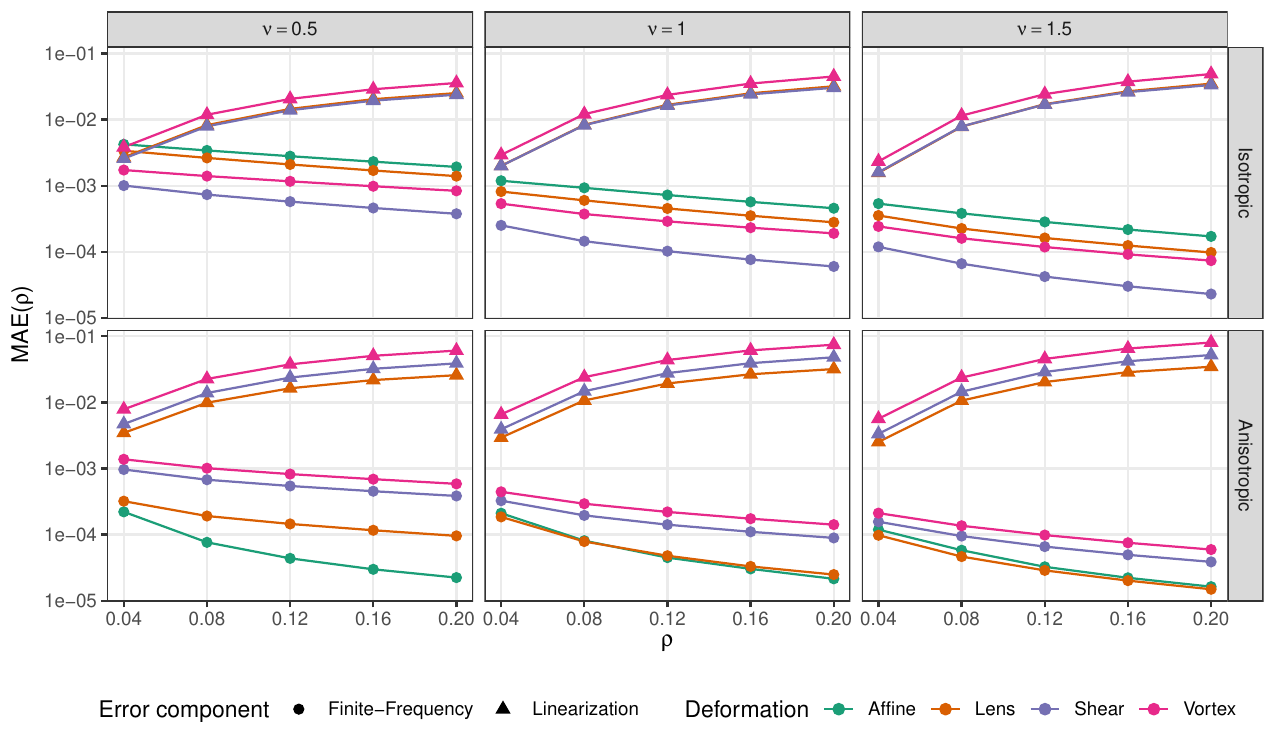}
    \caption{$\text{MAE}_{\text{freq}}(\rho)$ (circles) and $\text{MAE}_{\text{lin}}(\rho)$ (triangles) values for the isotropic Matérn cases (first row) and the anisotropic Matérn cases (second row) for varying values of smoothness (columns). Different colors denote different deformations.}
    \label{fig:cov_error_rho}
\end{figure}

The second diagnostic isolates the additional error due to the low-rank
compression of the local spectral amplitude matrix discussed in Section~\ref{sec:implementation}. For a fixed $\rho=0.12$, let $[\mathcal F]_{i\ell}=F_T(k_\ell,s_i)$, with $\ell=1,\dots,L$, and $s_i \in \mathcal{A}$. We compute the truncated SVD decomposition of $\mathcal F$ using the first $M$ components such that $M = \inf_M \left \{m: \frac{\sum_{m=1}^M \sigma_m^2}{\sum_{m} \sigma_m^2} \geq 0.99 \right \}$.
The zero frequency is stored separately. For each $s_0\in \mathcal{A}$, we compute the relative compression errors
$$
 e_{\mathrm{SVD}}^{\mathrm{rel}}(s_0)
    =
    \frac{
        \displaystyle
        \sum_{h\in\mathcal H_\rho(s_0)}
        \left|
        C_{L,M,s_0}(h)-C_{L,s_0}(h)
        \right|
    }{
        \displaystyle
        \sum_{h\in\mathcal H_\rho(s_0)}
        \left|C_{L,s_0}(h)\right|
    }.
$$
and in Table~\ref{tab:svd-compression} we report:
$$\operatorname{MAE}_{\mathrm{SVD}}^{\mathrm{rel}}
    =
    \frac{1}{|\mathcal A|}
    \sum_{s_0\in\mathcal A}
    e_{\mathrm{SVD}}^{\mathrm{rel}}(s_0).
$$
for each scenario. We use a relative error for this diagnostic because the goal is to measure the additional perturbation introduced by the SVD compression relative to the uncompressed finite-frequency covariance.

The results show that the low-rank compression introduces only a
modest additional covariance error. Across all nonlinear cases, the
displayed average relative errors do not exceed approximately $7\%$,
and the retained ranks remain small. The affine cases have rank one
and essentially zero compression error because their Jacobians do not vary
with the anchor location.
Among the nonlinear deformations, the lens and shear cases require
only two to four components, whereas the vortex requires rank five
for the isotropic baseline and ranks eight or nine for the
anisotropic baseline. The larger ranks in the latter cases are
consistent with the more complex spatial variation produced by the
interaction between the vortex deformation and the anisotropic
latent spectrum. Nevertheless, in these highest-rank cases, the
average relative compression error remains between approximately
$1.9\%$ and $2.1\%$. The largest displayed error, $7.0\%$, occurs for
the anisotropic lens case with $\nu=0.5$, thus for the area non-preserving deformation in the least smooth case.

\begin{table}[htbp]
\centering
\small
\caption{$\text{MAE}^{\text{rel}}_{\mathrm{SVD}}$ for each scenario. The selected rank $M$ is shown in
parentheses. }
\label{tab:svd-compression}
\begin{tabular}{llccc}
\toprule
Baseline & Deformation & $\nu=0.5$ & $\nu=1$ & $\nu=1.5$\\
\midrule
Isotropic & Affine & $<10^{-12}$ (1) & $<10^{-12}$ (1) & $<10^{-12}$ (1)\\
Isotropic & Lens   & 0.057 (3) & 0.056 (3) & 0.004 (4)\\
Isotropic & Shear  & 0.028 (2) & 0.027 (2) & 0.026 (2)\\
Isotropic & Vortex & 0.010 (5) & 0.010 (5) & 0.010 (5)\\
\midrule
Anisotropic & Affine & $<10^{-12}$ (1) & $<10^{-12}$ (1) & $<10^{-12}$ (1)\\
Anisotropic & Lens   & 0.070 (3) & 0.014 (4) & 0.014 (4)\\
Anisotropic & Shear  & 0.015 (3) & $<10^{-12}$ (4) & $<10^{-12}$ (4)\\
Anisotropic & Vortex & 0.021 (8) & 0.019 (9) & 0.019 (9)\\
\bottomrule
\end{tabular}
\end{table}

\subsection{Timing comparison}
\label{subsec:timing}
Now that we have established that the covariance errors are controlled even by compressing the spectral amplitude matrix, we investigate how efficient the proposed simulation is.
We compare two ways of producing repeated local outputs for
the isotropic Matérn baseline at $n=128$ and $\rho=0.12$.
The first method simulates a full stationary Matérn field by circulant
embedding, applies the deformation by interpolation, and extracts the requested
local neighborhoods. This is similar to the global deformation idea of \cite{kleiber2016}.
The second method precomputes the SVD-compressed local
spectral engine and then simulates only the requested neighborhoods. Both
methods reuse their setup across $100$ repeated draws, and the table reports
the total time of simulation. The global
method has almost constant draw cost as the number of requested anchors changes,
whereas the local method scales with the number of anchors and the nonlinearity of the deformation. We report the speedup value of the proposed local method compared to the global one. Values in parentheses larger than one indicate that the SVD local method is faster than global warping, and vise versa, values smaller than one indicate that global warping is faster.

We show the timing results in
Table~\ref{tab:timing-iso-affine-vortex}. The local simulator is most
effective when only a small number of anchor neighborhoods is required
and the Matérn field is sufficiently smooth. For the affine
deformation, the local method is faster than global warping for all
anchor counts when $\nu=1$ or $\nu=1.5$, whereas for $\nu=0.5$ it is
faster only with four anchors. For the vortex deformation, it is faster
with four anchors when $\nu=1$ or $\nu=1.5$, and retains a modest
advantage with eight anchors when $\nu=1.5$. For rougher fields or
larger collections of anchor neighborhoods, global simulation can be
faster.

\begin{table}[htbp]
\centering
\small
\caption{Timing results in seconds for the simulation of $100$ local random fields. We report in parenthesis the speedup factor compared to glabal deformation.}
\label{tab:timing-iso-affine-vortex}
\begin{tabular}{llccc}
\toprule
Deformation & $|\mathcal{A}|$
& $\nu=0.5$ & $\nu=1$ & $\nu=1.5$\\
\midrule
Affine & 4
& 1.99 $(\times 1.52)$
& 0.38 $(\times 8.08)$
& 0.14 $(\times 20.38)$\\
Affine & 8
& 3.93 $(\times 0.71)$
& 0.74 $(\times 4.19)$
& 0.28 $(\times 10.18)$\\
Affine & 16
& 8.59 $(\times 0.34)$
& 1.69 $(\times 1.85)$
& 0.62 $(\times 4.50)$\\
\midrule
Vortex & 4
& 2.45 $(\times 0.91)$
& 0.49 $(\times 4.80)$
& 0.19 $(\times 11.56)$\\
Vortex & 8
& 298.62 $(\times 0.01)$
& 4.04 $(\times 0.57)$
& 1.91 $(\times 1.22)$\\
Vortex & 16
& 695.78 $(\times 0.003)$
& 9.38 $(\times 0.25)$
& 4.46 $(\times 0.50)$\\
\bottomrule
\end{tabular}
\end{table}

The comparison with global warping should be interpreted as a benchmark against a setting in which the full deformation map $T$ is available. When the global deformation map is not available, but local Jacobian information or local spectral anisotropy can be specified or estimated, the proposed method remains a viable simulation strategy, whereas direct global warping cannot be applied.

\section{Optical-flow-informed local geometry in cardiac deformations}
\label{sec:acdc-application}

The application illustrates the value of exploring local second-order geometry in a nontrivial
setting. We focus on the data from the Automated Cardiac Diagnosis Challenge (ACDC) \citep{bernard2018}. The dataset presents sequences of cardiac magnetic resonance images for different patients. The objective of \cite{bernard2018} was to classify patients by diagnosis using image analysis. As they report, despite some patients belonging to different groups, the images look similar, and with clinical summaries like ventricular volumes, ejection fractions, myocardial mass, and wall thickness, it was difficult to tell them apart. We propose to estimate local deformation geometry on specific anatomical areas as an exploratory analysis rather than classification.

\subsection{Data description and preprocessing}
\label{subsec:acdc-data-preprocessing}

We use the ACDC training set, consisting of $n=100$ patients in five groups of twenty: normal subjects (NOR), myocardial infarction with altered left-ventricular ejection fraction (MINF), dilated cardiomyopathy (DCM), hypertrophic cardiomyopathy (HCM), and abnormal right ventricle (ARV). For each patient, the end diastole (ED) and end systole (ES) three-dimensional image volumes, and their pixel labels are available, denoting the right-ventricular cavity, left-ventricular myocardium, and left-ventricular cavity.

For each patient, we select two images, one representing ED and one ES. Specifically, first the two-dimensional mid-ventricular ES is selected such that the slice contains all three anatomical labels and is closest to the center of the stack. The same slice index is used to select the image in the ED tensor. Thus, we obtain images $\text{Im}_{\text{ES}}$ and $\text{Im}_{\text{ED}}$, which are cropped using a bounding box obtained from the union of the ED and ES segmentation masks, then resized and normalized. Notice that the images take brightness scalar values.

In order for the second order geometry to be analyzed, we need a deformation map $T$. We estimate that via optical flow \citep{horn1981}. Optical flow measures the displacement of a scalar quantity on local windows, allowing to estimate the velocity direction and magnitude for said window. We use it to estimate a dense optical-flow field $u(s)=(u_1(s),u_2(s))$ from $\text{Im}_{\text{ED}}$ to $\text{Im}_{\text{ES}}$. Under the assumption that $\operatorname{Im}_{\rm ES}(s)
    \approx
    \operatorname{Im}_{\rm ED}(s-u(s))$, then 
\begin{equation}
\label{eq:of}
    T(s)=s-u(s),
    \qquad
    J_T(s)=I-\nabla u(s),
\end{equation}
where derivatives are computed by finite differences after smoothing the displacement field. 
Thus, $T(s)$ maps a location $s$ in the ES image to its corresponding
location in the ED image, and we treat each image pixel as a spatial location. 
The derivatives in $J_T(s)$ are computed by
finite differences after smoothing the estimated displacement field.

Our objective is to measure the deformation from ED to ES in some relevant locations. We focus on the deformation of portions of the myocardium. The ACDC labels define two anatomy-relative myocardial regions. Let $c_{\mathrm{LV}}$ and $c_{\mathrm{RV}}$ be the centroids of the left ventricle (LV) and right ventricle (RV) cavities on $\text{Im}_{\text{ES}}$. The vector $(c_{\mathrm{RV}}-c_{\mathrm{LV}})$ points from the LV cavity toward the RV cavity, and gives a reproducible reference direction for the RV-facing, or septal, side of the myocardium. Among pixels labeled as myocardium, the septal region is the angular sector within $60^\circ$ of this LV-to-RV direction. The lateral  region is the opposite myocardial sector, centered $180^\circ$ away around the $c_{\mathrm{LV}}$. Figure~\ref{fig:acdc-pipeline} gives a visual illustration of the myocardial regions of interest and the estimated optical flow.

\begin{figure}[htbp]
    \centering
    \includegraphics[width=\linewidth]{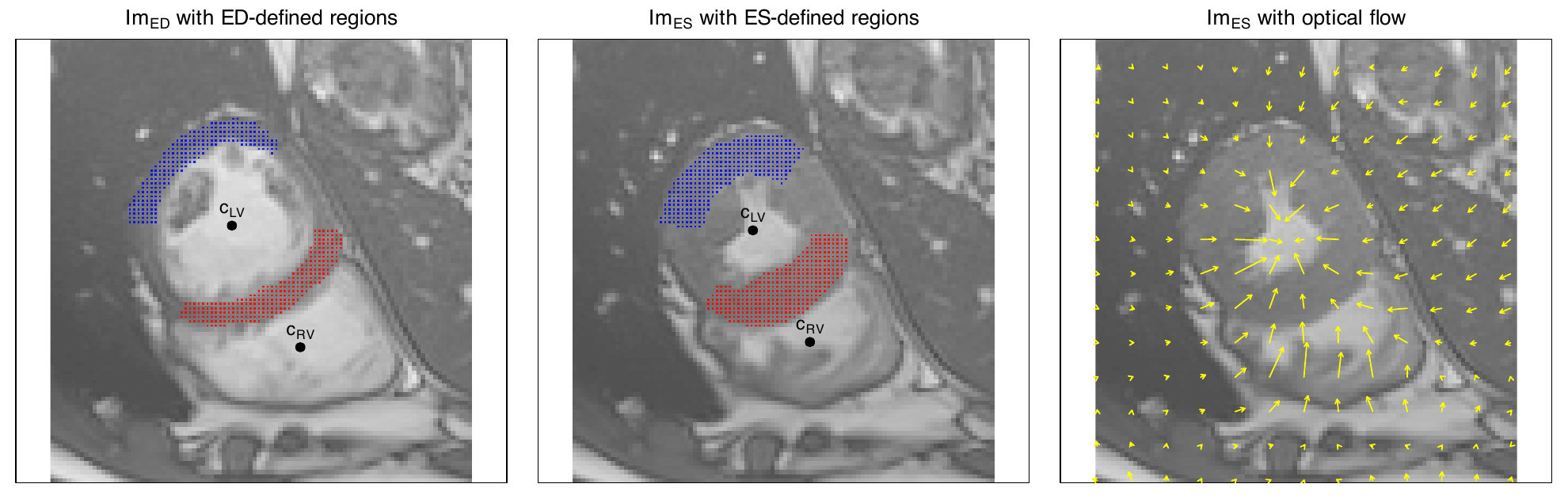}
    \caption{Illustration of the ACDC preprocessing for one representative patient. The first and second plots shows the selected ED and ES slices, respectively, with LV and RV centroids, and the septal (red) and lateral (blue) myocardial sectors. The third plot shows the estimated ED-to-ES optical-flow field (yellow).}
\label{fig:acdc-pipeline}
\end{figure}

For both sets of pixels in the septal and lateral regions, we use an isotropic Mat\'ern covariance as a neutral pre-deformation benchmark with no preferred orientation, so the anisotropy observed after deformation is attributable  to the Jacobian $J_T(s)$ only. Example~\ref{ex:Matern} shows that this Jacobian determines the local covariance through $G(s)=J_T(s)^\top J_T(s)$ in \eqref{eq:def_Matern_cov}. We ask whether deformation-derived summaries of local contraction geometry differ across the ACDC groups.

In particular, at each myocardial pixel $s$ in either region, we compute three summaries from $J_T(s)$. The first is $\log|\det J_T(s)|$, which measures local area change. This is a measure of purely compression and dilation, and does not give us any information on the direction of deformation. The second uses the eigenvalues $\lambda_1(s)\ge \lambda_2(s)>0$ of $G(s)$:
$$
    \log\eta(s)=\frac12\log\left(\frac{\lambda_1(s)}{\lambda_2(s)}\right).
$$
This measures the magnitude of local anisotropy, with $\log\eta(s)=0$ for a  spherical geometry. The third summary is the orientation of the principal local direction relative to the local radial direction. If $v_1(s)$ is the eigenvector associated with $\lambda_1(s)$ and $e_{\perp}(s)=\frac{s-c_{\mathrm{LV}}}{\|s-c_{\mathrm{LV}}\|}$
is the direction between $s$ and the centroid of the LV, then
$$
    \theta(s)=\cos^{-1}\left(|v_1(s)^\top e_{\perp}(s)|\right)\in[0^\circ,90^\circ],
$$
tells us if the main axis of the ellipses is aligned towards the center of the LV cavity. Specifically, values near $0^\circ$ indicate alignment with $e_{\perp}(s)$, the direction perpendicular to the myocardium, and values near $90^\circ$ indicate alignment tangential to the myocardium.

For each patient, we estimate the local Jacobian of the pixels in the septal and lateral regions. For each location of the two regions, we compute the three quantities of interest. Finally, we use the median of each one of these for the following analysis. Therefore, for both septal and lateral regions, we analyze $n=100$ representatives of each of the three features.

\subsection{Analysis and results}
\label{subsec:acdc-analysis-results}

For each of the six feature--region pair, the Kruskal--Wallis test \citep{kruskal1952} compares the rank distributions of the patient-level medians across the five ACDC groups. 
The Benjamini--Hochberg procedure \citep{benjamini1995} controls the false discovery rate across the feature--region tests, and gives the adjusted $p$-values, that we call the $q$-values. A significant test indicates group differences, but does not identify pairwise contrasts.

Table~\ref{tab:acdc-kw-tests} shows that the strongest group differences are directional. The largest effects are septal anisotropy, septal orientation, and lateral anisotropy. The local area-change summaries are also significant, but their $q$-values are several orders of magnitude larger. Thus, the main separation among groups is not only how much local area changes from ED to ES, but how directional the local deformation is and how this direction is aligned with cardiac anatomy.

\begin{table}[htbp]
\centering
\caption{Kruskal--Wallis tests results across ACDC groups using patient-level regional medians. We report the quantity tested (first column), the region (second column), the test statistic $H$ (third column), the $p$-value (fourth column) and the $q$-value (fifth column).}
\label{tab:acdc-kw-tests}
\begin{tabular}{llrrrr}
\toprule
Feature & Region & $H$ & $p$ & $q$  \\
\midrule
$\log\eta$ & Septum & 45.48 & 3.2$\times 10^{-9}$ & 1.9$\times 10^{-8}$  \\
$\theta_{\mathrm{rad}}$ & Septum & 42.99 & 1.0$\times 10^{-8}$ & 3.1$\times 10^{-8}$\\
$\log\eta$ & Lateral & 38.28 & 9.8$\times 10^{-8}$ & 2.0$\times 10^{-7}$ \\
$\log|\det J_T|$ & Septum & 17.45 & 2$\times 10^{-3}$ & 2$\times 10^{-3}$  \\
$\theta_{\mathrm{rad}}$ & Lateral & 14.43 & 6$\times 10^{-3}$ & 7$\times 10^{-3}$ \\
$\log|\det J_T|$ & Lateral & 10.41 & 3.4$\times 10^{-2}$ & 3.4$\times 10^{-2}$  \\
\bottomrule
\end{tabular}
\end{table}

Now that we have confirmed that there is a significant difference between diagnostic groups at the median level, we want to investigate how the groups differ. Figure~\ref{fig:acdc-top3-boxplots} shows the empirical distributions across patients of the top three quantities that the tests signalled as different across groups. We see that the strongest separation is not only between normal and pathological subjects, but among groups affected from different pathologies. Specifically, NOR, HCM and ARV show comparable empirical distributions, supporting the fact that none of these features alone is disease specific, and interpretation should remain descriptive.
In the septum, DCM has lower values of $\log\eta$ and values of $\theta$ that are centered around the lowest median and with high variability, suggesting a less directionally organized ED-to-ES deformation and a principal direction that is less consistently tangent to the myocardium. In the lateral wall, NOR and ARV have larger $\log\eta$ values, while MINF and DCM have lower values, suggesting weaker directional deformation in those groups. This complements the discussion in \cite{bernard2018}, where MINF and DCM are described as visually similar but different in local myocardial contraction behavior. Here, that distinction is expressed through local deformation geometry derived from the ED-to-ES motion field. 

\begin{figure}[htbp]
\centering
\includegraphics[width=\linewidth]{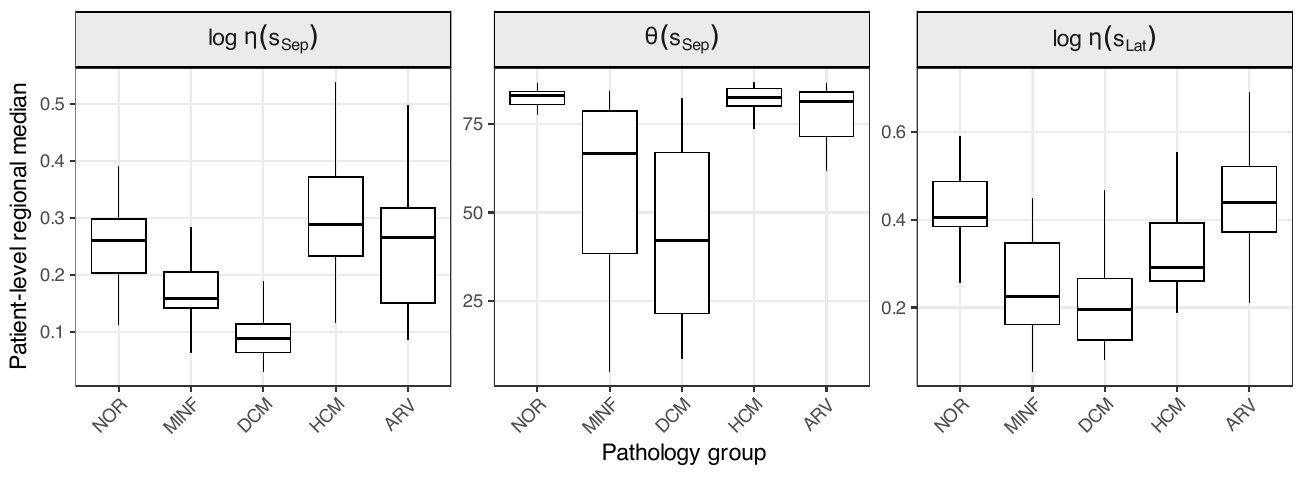}
\caption{Patient-level regional medians for the three strongest group differences in local deformation geometry: $\log\eta(s_{\mathrm{Sep}})$ (first panel), $\theta(s_{\mathrm{Sep}})$ (second panel), and $\log\eta(s_{\mathrm{Lat}})$ (third panel).}
\label{fig:acdc-top3-boxplots}
\end{figure}

\section{Discussion}
\label{sec:discussion}

This paper developed a local second-order analysis of spatial deformations. Starting from a stationary latent field deformed through a deterministic map, we studied how the deformation changes the covariance and how this change can be approximated locally through the Jacobian. The resulting covariance approximation comes with an explicit error bound, which separates the role of the covariance regularity from the curvature of the deformation. We then used the linearized covariance to derive a local spectral density, showing how the deformation changes local scale, orientation, and anisotropy while preserving the marginal variance. These results provide a link between deformation geometry, encoded by the Jacobian at a location, and its statistical expression via the local covariance and spectrum.

We made use of these results in two ways. First, we provided a simulation device for Gaussian random fields. The finite-frequency construction simulates a local Gaussian field as a sum of random Fourier components. The weights of these components are determined by the square-root local spectrum, so the resulting covariance matches the local covariance implied by the deformation. Moreover, a low-rank SVD gives a faster simulator with a controllable compression error. Overall, the numerical study confirmed that the method behaves consistently with the theoretical error decomposition and that the low-rank approximation can provide an accurate compact representation of the local spectral structure.

Second, we showed how the theoretical results can be used as exploratory tools for image-derived deformations. In the cardiac application, optical-flow estimates provided local Jacobians from which we summarized the determinant and eigenstructure of the local deformations. The analysis showed that differences across pathology groups are not captured only by the magnitude of expansion and compression, but also by directional features. 

This study opens several directions for further development. First, there are two sources of error or uncertainty that could be studied more deeply. The first is the approximation error introduced by the local linearization of the deformation. As the numerical results showed, this error is larger for highly nonlinear deformations and for less smooth latent covariance models. In such cases, higher-order local expansions could be considered, retaining curvature information from the deformation to reduce the bias of the local approximation. The second source is the uncertainty associated with an estimated deformation map. In many applications, the deformation and its Jacobian are not known exactly but estimated from data. Future work could propagate this uncertainty through the local covariance, local spectrum, and their derived summaries.

Second, this paper focused on univariate random fields. An extension to multivariate fields would require modeling not only the covariance across spatial locations, but also the cross-covariances among field components. In that setting, components may share a common deformation, have component-specific deformations, or be coupled through a joint deformation structure. The corresponding local theory would have to take into account multivariate dependence as well.

Finally, the framework could be extended to space-time modeling, where dependence changes not only across space but also through temporal evolution. A particularly relevant case is transport, since constant global advection can be viewed as the simplest affine deformation of space over time, in which the field is shifted rigidly. Recent work has made this model more realistic by replacing a single constant velocity with richer velocity structures and by introducing spectral damping to allow temporal decorrelation \citep{battagliola2026}. Extending the proposed local covariance and spectral analysis to such models would connect deformation geometry with nonseparable space-time dependence.

\bibliographystyle{apalike}
\bibliography{bib}

\end{document}